\documentclass{article}

\usepackage{calc}
\usepackage{amsmath}
\usepackage{amsxtra}
\usepackage{amssymb}
\usepackage{amsfonts}
\usepackage{amsthm}
\usepackage{amstext}
\usepackage{amsbsy}
\usepackage{amscd}
\usepackage[dvips]{color}

\usepackage[sans]{dsfont}

\usepackage[dvips]{graphicx}

\theoremstyle{plain}
\newtheorem{theorem}{Theorem}
\newtheorem{lemma}{Lemma}

\newtheorem{proposition}[lemma]{Proposition}
\theoremstyle{definition}

\newtheorem{definition}[lemma]{Definition}
\theoremstyle{remark}
\newtheorem{remark}[lemma]{Remark}

\newcommand{\R}{\mathds{R}}

\newcommand{\Z}{\mathds{Z}}
\newcommand{\defm}[1]{\emph{#1}}
\newcommand{\subeq}[2]{\mathord{\underbrace{\mathop{#1}}_{#2}}}

\newcommand{\trace}{\operatorname{tr}}

\newcommand{\find}{\operatorname{ind}}
\newcommand{\ran}{\operatorname{ran}}

\newcommand{\vx}{\vec\xi}

\newcommand{\myeqref}[1]{(\ref{#1})}
\newcommand{\myref}[1]{\ref{#1}}

\newcommand{\mylabel}[1]{\label{#1}}
\newcommand{\myeqlabel}[1]{\label{#1}}

\newcommand{\myquote}[1]{%
    \par\hspace{1cm}\parbox{\linewidth-2cm}{#1}\hspace{1cm}\par
}

\def\pmpaper{elling-liu-pmeyer}
\def\pmnewlabel#1#2{\expandafter\def\csname pmref-#1\endcsname{#2}}
\def\pmref#1{\csname pmref-#1\endcsname}
\def\pmeqref#1{(\csname pmref-#1\endcsname)}
\def\pmc#1{\cite[#1]{\pmpaper}}

\pmnewlabel{fig:weakstrong}{1}
\pmnewlabel{fig:wedgeflow}{2}
\pmnewlabel{fig:spolar-lines}{3}
\pmnewlabel{section:numerics}{1.2}
\pmnewlabel{fig:fullsol}{4}
\pmnewlabel{fig:frameorig}{4}
\pmnewlabel{th:elling-liu}{1}
\pmnewlabel{fig:numerics}{5}
\pmnewlabel{eq:techcond}{1.3.1}
\pmnewlabel{eq:prob1}{1.3.2}
\pmnewlabel{eq:prob2}{1.3.3}
\pmnewlabel{eq:prob3}{1.3.4}
\pmnewlabel{rem:weaksol}{1.3.1}
\pmnewlabel{eq:weak-ini}{1.3.5}
\pmnewlabel{fig:techcond}{6}
\pmnewlabel{section:ppf}{2}
\pmnewlabel{section:pf}{2.1}
\pmnewlabel{eq:rhodiv}{2.1.1}
\pmnewlabel{eq:mom}{2.1.2}
\pmnewlabel{eq:p-polytropic}{2.1.3}
\pmnewlabel{eq:v}{2.1.4}
\pmnewlabel{eq:rhoA}{2.1.5}
\pmnewlabel{eq:potflow-divform}{2.1.6}
\pmnewlabel{eq:rho}{2.1.7}
\pmnewlabel{eq:Dpiinv}{2.1.8}
\pmnewlabel{eq:potential-flow}{2.1.9}
\pmnewlabel{eq:cs-uspf}{2.1.10}
\pmnewlabel{section:sspf}{2.2}
\pmnewlabel{eq:psi-phi}{2.2.1}
\pmnewlabel{eq:rhoeq}{2.2.2}
\pmnewlabel{eq:chi-divform}{2.2.3}
\pmnewlabel{eq:chi}{2.2.4}
\pmnewlabel{eq:psi}{2.2.5}
\pmnewlabel{eq:css}{2.2.6}
\pmnewlabel{rem:symmetries}{2.2.1}
\pmnewlabel{eq:L}{2.2.7}
\pmnewlabel{eq:chijump}{2.3.1}
\pmnewlabel{eq:momjump}{2.3.2}
\pmnewlabel{eq:psijump}{2.3.3}
\pmnewlabel{eq:chitan}{2.3.4}
\pmnewlabel{eq:psitan}{2.3.5}
\pmnewlabel{eq:steady-continuity}{2.3.6}
\pmnewlabel{eq:chitan-z}{2.3.7}
\pmnewlabel{section:shocks}{2.4}
\pmnewlabel{eq:steady-rho}{2.4.1}
\pmnewlabel{eq:normal-rho}{2.4.2}
\pmnewlabel{lemma:srel-M}{2.4.1}
\pmnewlabel{eq:srel-M}{2.4.3}
\pmnewlabel{eq:g}{2.4.4}
\pmnewlabel{eq:dgdM}{2.4.5}
\pmnewlabel{eq:ddgddM}{2.4.6}
\pmnewlabel{eq:gM0}{2.4.7}
\pmnewlabel{eq:gMinf}{2.4.8}
\pmnewlabel{eq:cM}{2.4.9}
\pmnewlabel{eq:rhoM}{2.4.10}
\pmnewlabel{eq:cMpre}{2.4.11}
\pmnewlabel{lemma:srel}{2.4.2}
\pmnewlabel{eq:Mnd-asym}{2.4.12}
\pmnewlabel{eq:dMRdMLgen}{2.4.13}
\pmnewlabel{eq:dMRdML}{2.4.14}
\pmnewlabel{prop:shockv}{2.4.4}
\pmnewlabel{eq:DDspecial}{2.4.16}
\pmnewlabel{eq:DvndDvnu}{2.4.17}
\pmnewlabel{eq:zndznucomp}{2.4.18}
\pmnewlabel{lemma:movingnormal}{2.4.6}
\pmnewlabel{eq:DvndDsigma}{2.4.19}
\pmnewlabel{eq:DrhoDsigma}{2.4.20}
\pmnewlabel{section:shockpolar}{2.5}
\pmnewlabel{prop:shockpolar}{2.5.1}
\pmnewlabel{eq:DvdxDnva}{2.5.1}
\pmnewlabel{eq:DvdyDnva}{2.5.2}
\pmnewlabel{eq:blark}{2.5.3}
\pmnewlabel{fig:shockL}{7}
\pmnewlabel{prop:steady-shock-circle}{2.6.1}
\pmnewlabel{prop:vdzero}{2.6.2}
\pmnewlabel{fig:horvzero}{8}
\pmnewlabel{eq:bluaaa}{2.6.1}
\pmnewlabel{eq:vydeta}{2.6.2}
\pmnewlabel{eq:etaMbeta}{2.6.3}
\pmnewlabel{section:apriori}{3}
\pmnewlabel{lemma:alemma2}{3.1.1}
\pmnewlabel{eq:alemma2}{3.1.1}
\pmnewlabel{lemma:alemma}{3.1.2}
\pmnewlabel{eq:alemma}{3.1.2}
\pmnewlabel{eq:aineq}{3.1.3}
\pmnewlabel{lemma:interior-higher}{3.1.3}
\pmnewlabel{eq:lihcond}{3.1.4}
\pmnewlabel{eq:genDpsi}{3.1.5}
\pmnewlabel{eq:genDf}{3.1.6}
\pmnewlabel{prop:c-principle}{3.2.1}
\pmnewlabel{eq:rhomin2}{3.2.1}
\pmnewlabel{prop:interior-velocity}{3.3.1}
\pmnewlabel{prop:v-wall}{3.4.1}
\pmnewlabel{prop:vshock}{3.5.1}
\pmnewlabel{eq:propvx1}{3.5.1}
\pmnewlabel{eq:propvx-s11}{3.5.2}
\pmnewlabel{eq:blurb}{3.5.3}
\pmnewlabel{eq:rhos-psi12}{3.5.4}
\pmnewlabel{eq:wsys}{3.5.5}
\pmnewlabel{eq:wdet}{3.5.6}
\pmnewlabel{prop:L-minmax}{3.6.1}
\pmnewlabel{eq:rhoLb}{3.6.1}
\pmnewlabel{eq:rho-sstr}{3.6.2}
\pmnewlabel{eq:L2}{3.6.4}
\pmnewlabel{prop:density-shock}{3.7.1}
\pmnewlabel{fig:rhos11}{9}
\pmnewlabel{fig:rhotangent}{10}
\pmnewlabel{fig:shockstrength}{11}
\pmnewlabel{section:ellreg}{4}
\pmnewlabel{section:approach}{4.2}
\pmnewlabel{fig:frameR}{12}
\pmnewlabel{fig:frameL}{12}
\pmnewlabel{fig:regularized}{13}
\pmnewlabel{eq:chitt-exp}{4.2.1}
\pmnewlabel{eq:chitt-exp2}{4.2.2}
\pmnewlabel{eq:chitt-exp3}{4.2.3}
\pmnewlabel{section:parmset}{4.3}
\pmnewlabel{eq:deltant}{4.3.1}
\pmnewlabel{eq:dntcond}{4.3.2}
\pmnewlabel{eq:constlist}{4.3.3}
\pmnewlabel{def:Lambda}{4.3.1}
\pmnewlabel{eq:Ceta}{4.3.4}
\pmnewlabel{lemma:etax}{4.3.2}
\pmnewlabel{fig:etaL}{14}
\pmnewlabel{lemma:Gamma-connected}{4.3.3}
\pmnewlabel{eq:ueLbd}{4.3.5}
\pmnewlabel{def:weighted-hoelder}{4.4.1}
\pmnewlabel{fig:onion}{15}
\pmnewlabel{def:b}{4.4.2}
\pmnewlabel{def:fusp}{4.4.3}
\pmnewlabel{eq:Tt}{4.4.1}
\pmnewlabel{eq:regularity}{4.4.2}
\pmnewlabel{eq:sdef}{4.4.3}
\pmnewlabel{eq:shockwall}{4.4.4}
\pmnewlabel{eq:s-welldef}{4.4.5}
\pmnewlabel{eq:cornerregion}{4.4.6}
\pmnewlabel{eq:cornercone}{4.4.7}
\pmnewlabel{eq:rhoprep}{4.4.8}
\pmnewlabel{eq:rhomin}{4.4.9}
\pmnewlabel{eq:ellip}{4.4.10}
\pmnewlabel{eq:ellipC}{4.4.11}
\pmnewlabel{eq:oldnew}{4.4.12}
\pmnewlabel{eq:tilderho}{4.4.13}
\pmnewlabel{eq:faken}{4.4.14}
\pmnewlabel{eq:tildeL}{4.4.15}
\pmnewlabel{eq:itn-inner}{4.4.16}
\pmnewlabel{eq:itn-parabolic}{4.4.17}
\pmnewlabel{eq:itn-shock}{4.4.18}
\pmnewlabel{eq:itn-wall}{4.4.19}
\pmnewlabel{eq:pn}{4.4.20}
\pmnewlabel{eq:lip}{4.4.21}
\pmnewlabel{eq:partan}{4.4.22}
\pmnewlabel{eq:parnor}{4.4.23}
\pmnewlabel{eq:horvel}{4.4.24}
\pmnewlabel{eq:vertvel}{4.4.25}
\pmnewlabel{eq:leftvel}{4.4.26}
\pmnewlabel{eq:shocknormal}{4.4.27}
\pmnewlabel{eq:ndobb}{4.4.28}
\pmnewlabel{eq:Gb}{4.4.29}
\pmnewlabel{rem:fp}{4.4.4}
\pmnewlabel{rem:reflection}{4.4.5}
\pmnewlabel{prop:Lxn-iso}{4.4.6}
\pmnewlabel{eq:eqlin}{4.4.30}
\pmnewlabel{eq:shocklin}{4.4.31}
\pmnewlabel{eq:parlin}{4.4.32}
\pmnewlabel{eq:walllin}{4.4.33}
\pmnewlabel{prop:pn-uqcont}{4.4.7}
\pmnewlabel{prop:fusp-topology}{4.4.8}
\pmnewlabel{def:it}{4.4.9}
\pmnewlabel{rem:mutrans}{4.5.1}
\pmnewlabel{prop:regularity}{4.5.2}
\pmnewlabel{prop:fp-regularity}{4.5.2}
\pmnewlabel{eq:sLip}{4.5.1}
\pmnewlabel{eq:sregu}{4.5.2}
\pmnewlabel{eq:reguint}{4.5.3}
\pmnewlabel{eq:regu2}{4.5.4}
\pmnewlabel{eq:lipode}{4.5.5}
\pmnewlabel{prop:it-continuous-compact}{4.5.3}
\pmnewlabel{section:L-control}{4.6}
\pmnewlabel{prop:Lbounds}{4.6.1}
\pmnewlabel{eq:Leps}{4.6.1}
\pmnewlabel{fig:pararc}{16}
\pmnewlabel{section:parcs}{4.7}
\pmnewlabel{section:c-pararc}{4.7}
\pmnewlabel{eq:Lsimple}{4.7.1}
\pmnewlabel{eq:Ltchi}{4.7.2}
\pmnewlabel{eq:Lnchi}{4.7.3}
\pmnewlabel{eq:psitautau}{4.7.5}
\pmnewlabel{eq:refp}{4.7.6}
\pmnewlabel{eq:cc-arc}{4.7.7}
\pmnewlabel{eq:h0}{4.7.8}
\pmnewlabel{eq:refk}{4.7.9}
\pmnewlabel{eq:pphi}{4.7.10}
\pmnewlabel{eq:kphi}{4.7.11}
\pmnewlabel{eq:qphi}{4.7.12}
\pmnewlabel{eq:thetaphi}{4.7.13}
\pmnewlabel{prop:pararc}{4.8.1}
\pmnewlabel{eq:rhoP}{4.8.1}
\pmnewlabel{eq:vP}{4.8.2}
\pmnewlabel{eq:thetasector}{4.8.3}
\pmnewlabel{eq:pcontrol}{4.8.4}
\pmnewlabel{eq:ccontrol}{4.8.5}
\pmnewlabel{eq:qcontrol}{4.8.6}
\pmnewlabel{eq:pphi-isen}{4.8.7}
\pmnewlabel{eq:plower}{4.8.8}
\pmnewlabel{eq:pupper}{4.8.9}
\pmnewlabel{fig:thetasector}{17}
\pmnewlabel{section:cornersmoving}{4.9}
\pmnewlabel{eq:zdzuLone}{4.9.1}
\pmnewlabel{eq:zdeta}{4.9.2}
\pmnewlabel{eq:cp1}{4.9.3}
\pmnewlabel{eq:cp2}{4.9.4}
\pmnewlabel{eq:cp3}{4.9.5}
\pmnewlabel{eq:zydpre}{4.9.6}
\pmnewlabel{eq:vyd-eta}{4.9.7}
\pmnewlabel{eq:vyd-eta-positive}{4.9.8}
\pmnewlabel{eq:zxdpre}{4.9.9}
\pmnewlabel{eq:pc-zphi}{4.9.10}
\pmnewlabel{eq:petaineq}{4.9.11}
\pmnewlabel{eq:pc-csq}{4.9.12}
\pmnewlabel{eq:keta}{4.9.13}
\pmnewlabel{eq:pplus}{4.9.14}
\pmnewlabel{eq:kplus}{4.9.15}
\pmnewlabel{eq:qetaB}{4.9.16}
\pmnewlabel{eq:qqq}{4.9.17}
\pmnewlabel{eq:qqeta}{4.9.18}
\pmnewlabel{eq:qeta}{4.9.19}
\pmnewlabel{eq:theta-phibar-plus}{4.9.20}
\pmnewlabel{eq:theta-phibar-minus}{4.9.21}
\pmnewlabel{section:lowerbounds}{4.10}
\pmnewlabel{prop:etaa-lowerbound}{4.10.1}
\pmnewlabel{fig:cmax}{18}
\pmnewlabel{prop:cbar}{4.10.2}
\pmnewlabel{eq:cbar}{4.10.1}
\pmnewlabel{prop:psi-axi}{4.10.3}
\pmnewlabel{eq:a}{4.10.3}
\pmnewlabel{eq:aexp}{4.10.4}
\pmnewlabel{prop:vyd-crit}{4.10.5}
\pmnewlabel{prop:etaa-upperbound}{4.10.6}
\pmnewlabel{section:densitycontrol}{4.11}
\pmnewlabel{prop:rho}{4.11.1}
\pmnewlabel{fig:shocktanarg}{19}
\pmnewlabel{section:v-control}{4.12}
\pmnewlabel{prop:ny}{4.12.1}
\pmnewlabel{prop:vx}{4.12.1}
\pmnewlabel{prop:vy}{4.12.1}
\pmnewlabel{eq:chitshock}{4.12.1}
\pmnewlabel{eq:s1x0}{4.12.2}
\pmnewlabel{prop:oblique-corner}{4.13.1}
\pmnewlabel{eq:gpS}{4.13.2}
\pmnewlabel{prop:fp-boundary}{4.13.2}
\pmnewlabel{section:ls}{4.14}
\pmnewlabel{prop:unperturbed-unique}{4.14.1}
\pmnewlabel{fig:unperturbed}{20}
\pmnewlabel{prop:unperturbed-index}{4.14.3}
\pmnewlabel{eq:paralin}{4.14.1}
\pmnewlabel{eq:fdeta}{4.14.2}
\pmnewlabel{eq:fdnor}{4.14.3}
\pmnewlabel{eq:shockl}{4.14.4}
\pmnewlabel{prop:probell}{4.15.1}
\pmnewlabel{section:entireflow}{4.16}
\pmnewlabel{fig:epslimit}{21}
\pmnewlabel{eq:interior-eps}{4.16.1}
\pmnewlabel{eq:para-chi-eps}{4.16.2}
\pmnewlabel{eq:para-rho-eps}{4.16.3}
\pmnewlabel{eq:para-nablachi-eps}{4.16.4}
\pmnewlabel{eq:shock1-eps}{4.16.5}
\pmnewlabel{eq:shock2-eps}{4.16.6}
\pmnewlabel{eq:cornerdist-eps}{4.16.8}
\pmnewlabel{eq:lip-eps}{4.16.9}
\pmnewlabel{eq:cka-eps}{4.16.10}
\pmnewlabel{eq:eps-weak}{4.16.11}
\pmnewlabel{eq:zero-weak}{4.16.12}
\pmnewlabel{eq:zero-cont}{4.16.13}
\pmnewlabel{eq:zero-cont-rhov}{4.16.14}
\pmnewlabel{rem:notvoid}{4.16.1}
\pmnewlabel{rem:structure}{4.16.2}
\pmnewlabel{fig:corner}{22}
\pmnewlabel{section:corner}{5.1}
\pmnewlabel{prop:corner}{5.1.1}
\pmnewlabel{eq:Gammaregu}{5.1.1}
\pmnewlabel{eq:phithetabd}{5.1.2}
\pmnewlabel{eq:uLip}{5.1.3}
\pmnewlabel{eq:u}{5.1.4}
\pmnewlabel{eq:Gammacond}{5.1.5}
\pmnewlabel{eq:coeffnorm}{5.1.6}
\pmnewlabel{eq:Ce}{5.1.7}
\pmnewlabel{eq:ndob}{5.1.8}
\pmnewlabel{eq:G}{5.1.9}
\pmnewlabel{eq:Colambda}{5.1.10}
\pmnewlabel{eq:ghregu}{5.1.11}
\pmnewlabel{eq:Ctabeta}{5.1.12}
\pmnewlabel{eq:w-interior}{5.1.13}
\pmnewlabel{eq:wboundary}{5.1.14}
\pmnewlabel{eq:ddudg}{5.1.15}
\pmnewlabel{eq:Av}{5.1.16}
\pmnewlabel{eq:gGammaf}{5.1.17}
\pmnewlabel{prop:ztrans}{5.2.1}
\pmnewlabel{eq:Dcond}{5.2.1}
\pmnewlabel{rem:shmor}{5.2.2}

\begin{document}

\title{Counterexamples to the sonic criterion}
\author{Volker Elling}
\date{}	

\maketitle

\begin{abstract}
    We consider self-similar (pseudo-steady) shock reflection at an oblique wall. 
    There are three parameters: wall corner angle, Mach number, angle of incident shock.
    Ever since Ernst Mach discovered the irregular reflection named after him, 
    it has been an open problem to predict precisely for what parameters 
    the reflection is regular. Three conflicting proposals, the detachment, sonic and von Neumann 
    criteria, have been studied extensively without a clear result.

    We demonstrate that the sonic criterion is not correct.
    We consider polytropic potential flow and prove that there is an 
    open nonempty set of parameters that admit a global regular reflection
    with a reflected shock that is \emph{transonic}. 

    We also provide a clear physical reason: the flow type (sub- or supersonic)
    is not decisive; instead the reflected shock type (weak or strong)
    determines whether structural perturbations decay towards the reflection point.
\end{abstract}

\parindent=0cm%
\parskip=0cm%

\parindent=0cm%
\parskip=\baselineskip%

\parindent=0cm%
\parskip=\baselineskip%

\section{Introduction}

\subsection{The transition problem in shock reflection}

\mylabel{section:refl}

\begin{figure}
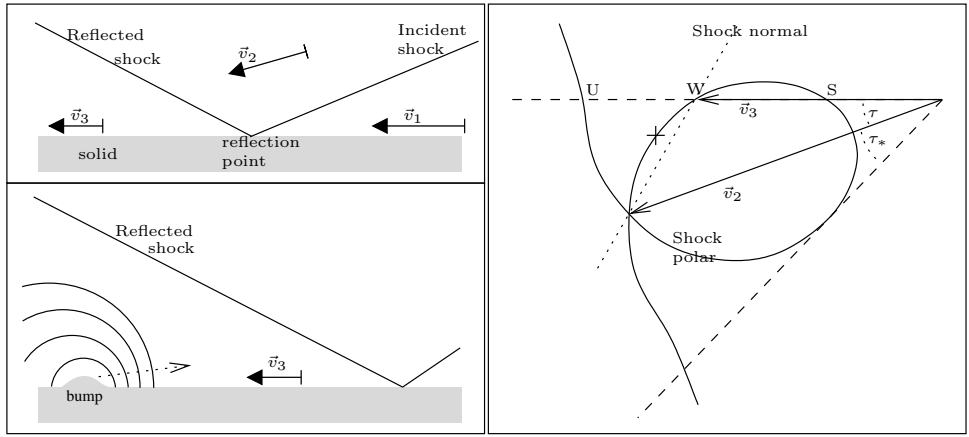

\input{locrr.pstex_t}%
\input{polar.pstex_t}%
\caption{Left top: local RR. 
Left bottom: subsonic case.
Right: fixed $\vec v_2$; each steady shock produces
one $\vec v_3$ on the curve (shock polar, symmetric across $\vec v_2$; shock normal $\parallel\vec v_2-\vec v_3$). 
For $|\tau|<\tau_*$, three shocks satisfy $\tau=\measuredangle(\vec v_2,\vec v_3)$: strong-type (S),
weak-type (W) and expansion (U; unphysical). W are transonic right of $+$, supersonic left.}
\mylabel{fig:locrr-left}%
\mylabel{fig:spolar-right}%
\end{figure}

\begin{figure}
\input{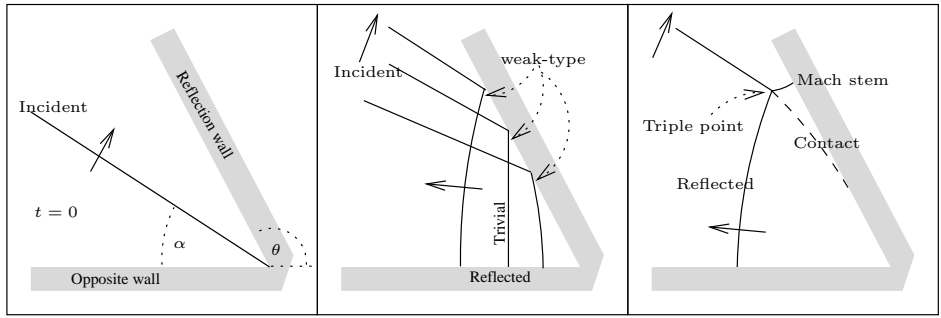}%
\caption{Left: initial data (from dotted area of Figure \myref{fig:experiment} second right).
Center: RR; we construct perturbations of the trivial case. Right: SMR}%
\label{fig:bigger90}%
\end{figure}

\begin{figure}
\input{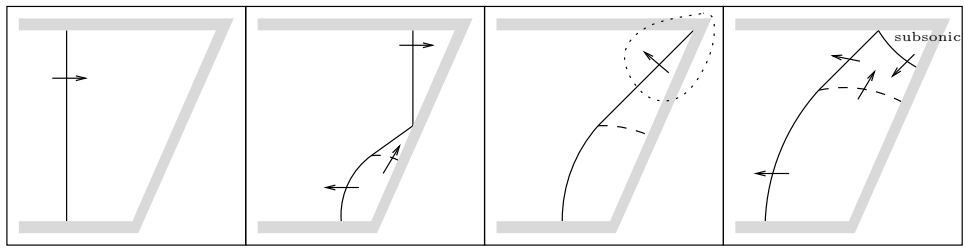}%
\caption{A shockwave breaks into a supersonic RR at the lower corner;
the reflected shock breaks in the upper corner to produce our kind of subsonic RR for
some time.}%
\mylabel{fig:experiment}%
\end{figure}

Reflection of an incident shock from a solid wedge is a classical problem of gas dynamics. 
It has been studied extensively by 
Ernst Mach \cite{mach-wosyka,krehl-geest} and John von Neumann \cite{neumann-1943}, among others.

Most commonly, reflection is studied in \emph{steady} inviscid 
polytropic\footnote{equation of state $p=(\gamma-1)\rho e$,
$e$ internal energy per mass, $\gamma\in(1,\infty)$}
compressible flow, for example when shocks in a nozzle are reflected from the walls.
The reflections can be classified roughly into \defm{regular} and \defm{irregular reflections};
see \cite{ben-dor-book} or \cite[Figure 1]{ben-dor-shockwaves2006} 
for a more detailed discussion. In either type, an \defm{incident shock} impinges on a solid surface.
In regular reflection (RR), the incident shock reaches a \defm{reflection point} on the surface, 
continuing as a \defm{reflected shock} (see Figure \myref{fig:locrr-left} top left).

In \defm{irregular reflections} (IRR), incident and reflected shock are connected by a more or less complex 
interaction pattern which in turn connects
to the solid surface by a third shock, called \defm{Mach stem}.
The most important irregular reflections are double, complex and single 
(see Figure \myref{fig:bigger90} right) Mach reflection (MR); additional types have been 
discussed \cite{guderley,vasiliev-kraiko,hunter-brio,hunter-tesdall,skews-ashworth}.

Some incident shocks allow more than one type of reflection. 
Assuming uniqueness for the problem at hand, only one of them can be extended to a 
global solution: a solution in the \emph{entire} domain, satisfying all boundary and far-field conditions.
A long-standing open question is to find the exact criterion that determines whether the solution is RR.

Among the criteria for appearance of RR that have been proposed (see \cite[Section 1.5]{ben-dor-book}), 
two are most important.
The \defm{detachment criterion} states that global RR appears generically whenever local RR is possible.

A physical argument motivates the second criterion: for a straight wall, 
all local RR and MR are trivially global solutions. But some of them could be unstable 
under perturbations, for example a bump in the $3$-sector wall (Figure \ref{fig:locrr-left} left bottom).
If so, then information is transmitted from the bump to the reflection point\footnote{this is known as 
\emph{information condition} or \defm{information argument}}.
For weak waves\footnote{but sufficiently strong shock waves can travel upstream against a supersonic flow}
that is possible if and only if the $3$-sector is \emph{subsonic}\footnote{the other sectors are always supersonic}. Hence the \defm{sonic criterion}: 
global RR appears generically if there is a \emph{supersonic} local RR, but not otherwise. 
(Each criterion can also be formulated in other, slightly different ways.)

\subsection{Weak- and strong-type}

The velocity $\vec v_2$ in the $2$-sector in Figure \myref{fig:locrr-left} forms an angle $\tau$ with the wall; 
the reflected shock must turn this
velocity by $\tau$ so that $\vec v_3$ is parallel to the wall, satisfying a slip boundary condition.

Keep the $2$-sector data fixed while rotating the reflected shock in the reflection point.
This yields a one-parameter family of velocities $\vec v_3$, forming a curve called \defm{shock polar} 
(see Figure \myref{fig:spolar-right} right). For \emph{admissible} shocks, $|\tau|$ cannot exceed $\tau_*$, 
the \defm{critical angle}, 
which is a function of the Mach number $M_2$ and $\gamma$. 

Throughout this paper we focus on polytropic equations of state
so that the admissible part of the shock polar is strictly convex.

If the angle $\tau$ between wall and $\vec v_2$ is bigger than
$\tau_*$, then local RR is theoretically impossible. If $\tau=\tau_*$, there is exactly one reflected shock,
called \defm{critical-type}. For $\tau<\tau_*$ however there
are \emph{two}, called \defm{weak-type}\footnote{The weak-type shock is relatively weaker than the strong-type shock, but
their absolute strength can be arbitrarily small or large, so we prefer to use \emph{-type}.} and \defm{strong-type}. 
We encounter another major issue in reflection:
which of these two should occur? \cite{elling-liu-pmeyer} have discussed this question for a different problem.

We call shocks \defm{transonic} if the downstream side is subsonic, 
\defm{supersonic} if both sides are supersonic. The weak shock is transonic for $\tau>\tau_s$ for some threshold
$\tau_s<\tau_*$, supersonic otherwise; the strong-type shock is always transonic.
In this paper we consider only transonic RR.

\subsection{Self-similar reflection problems}

Some variants of the reflection problem are 
\defm{self-similar}\footnote{also called \defm{quasi-steady} or \defm{pseudo-steady}} 
flow rather than steady.

In self-similar flow, density and velocity are functions of the \emph{similarity coordinates} 
$(\xi,\eta)=(x/t,y/t)$ rather than $x,y$. Patterns grow linearly in time, with $t\downarrow 0$ corresponding to ``zooming infinitely far away''
whereas $t\uparrow+\infty$ is like ``zooming into the origin'' or ``scaling up''. 
Here inviscid models are easily justified because any flow feature eventually grows beyond the 
length scale where dissipate or kinetic phenomena matter\footnote{unless these small-scale phenomena trigger 
large-scale effects like turbulence, boundary layer separation etc.}.
Self-similar reflections occur naturally in many 
experiments (see Figure \myref{fig:experiment}, \cite{henderson-etal,ben-dor-shockwaves2006}).

We consider three parameters (see Figure \ref{fig:bigger90} left): 
$M_1$, the $1$-sector Mach\newcounter{fnmach}\setcounter{fnmach}{\value{footnote}}
are defined number, $\alpha$, clockwise angle from opposite wall to incident shock, and 
$180^\circ-\theta$, clockwise angle from opposite wall to reflection wall.
The opposite wall passes\footnote{to satisfy a slip condition on the opposite wall} through $\vec\xi=\vec v_2$. 
Mach number and velocity are defined for an observer traveling in the reflection point.

For $t\downarrow 0$ this yields the initial 
data\footnote{If the incident shock forms a right angle to the upstream wall, this problem is familiar 
\cite{chen-feldman-selfsim-journal,elling-rrefl}. Note that the nonvertical cases also arise from certain $t<0$ flows;
in particular they can arise in simple experiments like Figure \myref{fig:experiment}}
seen in Figure \myref{fig:bigger90} left. Depending on $\theta$ either RR or MR appear.

If we choose the opposite wall perpendicular to the reflected shock, then local RR extends to a global \defm{trivial RR} 
(see Figure \myref{fig:bigger90} center).

\begin{figure}
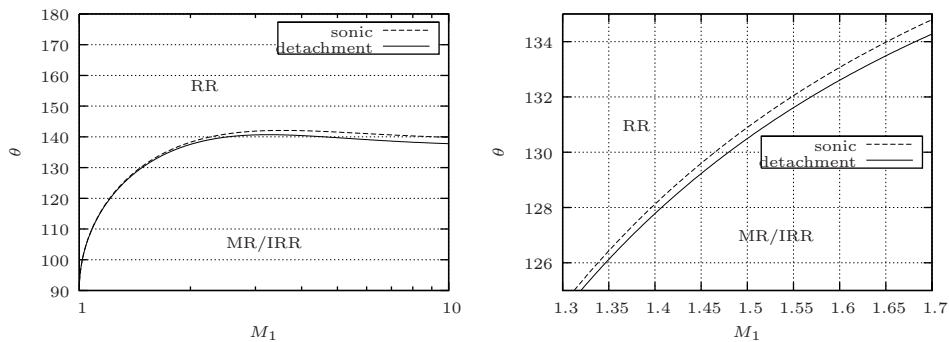

\input{rr-hor-potf.pstex_t}
\input{rr-hor-detail-potf.pstex_t}
\caption{Left: $\theta_d,\theta_s$ for $\gamma=7/5$ potential flow 
    and $\alpha=0$ in Figure \myref{fig:bigger90}.
    Right: detail.}
\label{fig:horrefl}
\end{figure}

\subsection{Transition}

The two transition criteria specify transition angles $\theta_d$ (detachment) and 
$\theta_s>\theta_d$ (sonic) depending on $L_1$, $\alpha$, $\gamma$.
Global RR is predicted for larger $\theta$ and IRR for smaller $\theta$. 
Figure \myref{fig:horrefl} compares the two criteria in the case of 
$\gamma=7/5$ polytropic potential flow.

To quote \cite{ben-dor-book}: 
\myquote{%
    ``For this reason it is almost impossible to distinguish experimentally between the sonic and detachment
    criteria.'' 
}
Experimental and numerical accuracy are affected by 
viscosity/heat conduction\footnote{Observations (e.g.\ \cite{van-dyke} p.\ 142f)
    agree with our analytical solutions, so inviscid models are clearly suitable.
    Experiments \cite{henderson-etal} show that,
    although viscous/boundary layer effects can have a transient 
    effect on the transition $\theta$, for sufficiently large times the 
    transition is close to the inviscid predictions $\theta_s,\theta_d$.},
non-equilibrum effects, turbulence, surface roughness and other systematic errors 
as well as noise. 
The interaction of physical or numerical boundary layers 
with RR causes \defm{spurious Mach stems} \cite[Figure 7a]{colella-woodward-review}
that make it look like MR, however boundaries can be avoided by reflection into an interior problem.

Although the question has remained open, the sonic condition appears to have been favored by many 
researchers (including the author), 
at least for small $M_1$. As the recent survey \cite{ben-dor-shockwaves2006} states, 
\myquote{%
    ``[...] the {[criterion]} which best agrees with pseudo-steady shock tube experimental data [...]
    suggests that in pseudo-steady flows RR terminates when the flow behind
    the reflection point, R [...] becomes sonic in a frame of reference attached to R.''
}

It should be noted that these quotes refer only to the classical case $\alpha=90^\circ$ and hence
$\theta<90^\circ$, i.e.\ \emph{vertical} incident shock (see Figure myref{fig:experiment} second left). 
Here we consider some cases with $\theta>90^\circ$ and $\alpha<90^\circ$ because they can be solved
by linearization around trivial RR, i.e.\ small-data techniques. 
The classical case requires a large-data approach as in \cite{elling-liu-pmeyer}; 
this will be subject of future research.
However, the nature of the question is the same in all cases: how is local RR affected by 
various kinds of perturbation. The classical perturbation occurs naturally in some experiments, but 
there is no other reason to favor it.

\subsection{Results}

We prove, for the self-similar reflection problem modeled with
potential flow, that the sonic criterion is not universally correct.
We use the following formulation\footnote{This version is a weak as possible, 
    by considering ``generic'' instead of all, and by requiring
    structural stability.}:
\myquote{%
    \emph{Generic} local transonic RR cannot extend into structurally stable global RR.
}

Instead, Theorem \myref{th:sonic-wrong} shows:
\myquote{%
    1. Trivial weak-type transonic RR is structurally stable.
}
In particular, the parameter space has an open nonempty --- hence generic --- subset 
with extendable local RR.

\addtocounter{footnote}{1}
\footnotetext{Here we mean decay in space, not in time.}
More importantly, we identify a physical reason for the failure of the sonic criterion.
\newcounter{fndecay}\setcounter{fndecay}{\value{footnote}}
The information argument (see above) indeed goes a long way towards the correct answer.
But interestingly, it is too restrictive in a subtle way:
\myquote{%
    2. For \emph{weak}-type transonic reflections, downstream perturbations \emph{can} 
    reach the reflection point, but they \emph{decay to zero}\footnotemark[\value{fndecay}] in the process.
}
This suggests that --- although a proof is given only for particular parameters --- the sonic criterion is 
incorrect for \emph{most}, if not all, parameters, in particular including the classical
case $\alpha=90^\circ$.

We demonstrate the principle for a particular variant of the reflection problem in potential flow. 
However, it will become clear
during the course of the proof of Theorem \myref{th:sonic-wrong} that gradient regularity near the reflection point
is a local property of elliptic PDE and their boundary conditions in a domain corner. Hence 
the same principle applies to other steady or self-similar variants. 
Moreover, the same regularity effect occurs in isentropic and non-isentropic Euler flow.

A wider range of parameters, Euler flow and the validity of the detachment criterion will 
be discussed in separate articles.
A third important criterion, the \defm{von Neumann criterion}\footnote{also referred to as 
\defm{mechanical equilibrum criterion} in some contexts}, does not apply at all in potential 
flow\footnote{Even in Euler flow it applies only for large $M_1$, for example $M_1>2.2...$ 
for $\gamma=7/5$.}.

\cite{elling-liu-pmeyer} previously provided a rigorous construction of supersonic weak-type reflections 
in a different problem. The techniques in this article are apparently sufficient to extend the construction 
to \emph{transonic} cases.

Considerations analogous to the sonic criterion have also been used in studying the transitions between
different types of Mach reflection (see the survey in \cite{ben-dor-shockwaves2006} for details). 
Our findings suggest 
modifications to these proposals as well, by replacing subsonic-ness with weak-type in some way.

\subsection{Other remarks}

Many articles have considered \emph{dynamic} 
stability\footnote{stability under perturbations to the initial data}, trying to show that at the linear 
level weak-type reflections are stable while strong-type are not. 
However, numerical calculations \cite[Figure 3]{elling-liu-rims05} suggest that both types are 
dynamically stable.

\cite{hornung-weakstrong} has previously proposed a plausibility argument, 
based on pressure changes, for stability
of weak-type transonic (and instability of strong-type) shocks. \cite{chen-zhang-zhu} 
show existence and structural
stability of \emph{supersonic} reflections from a wedge.
\cite{chen-feldman-selfsim-journal} have constructed global supersonic RR for 
$\alpha=90^\circ$ and $\theta\approx90^\circ$ as exact solutions of self-similar potential flow. \cite{elling-rrefl}
shows existence of global supersonic RR for a range of parameters that includes, in some cases, all $\theta>\theta_s$, proving
that criteria more restrictive than sonic cannot be universally correct.

It has been proposed that both RR and MR may occur for the same parameters in steady flow, with hysteresis effects
when parameters are changed (see e.g.\ \cite{ben-dor-ivanov-vasilev-elperin}). In self-similar flow this would
amount to non-uniqueness for an initial-value problem. 
Indeed, \cite{elling-hyp2004,elling-nuq-journal} has found a set
of initial data for the 2d Euler equations (both isentropic and non-isentropic) that appears to have two
solutions, one theoretical, the other clearly different and observed in all numerical calculations\footnote{In addition it is shown that the Godunov scheme can converge to either solution, depending on the grid.}
For isentropic Euler, a rigorous proof of a different non-uniqueness example has recently been proposed
\cite{de-lellis-szekelyhidi}. 

However, both results depend strongly on vorticity; uniqueness for the potential flow Cauchy problem
is still expected and hysteresis is unlikely except as a transient phenomenon.

\section{Self-similar potential flow}

Here we prove technical results which are not previously available in the literature.

\subsection{Equations}

2d isentropic Euler flow is a PDE system for a density field $\rho$ and velocity field $\vec v$,
consisting of the continuity equation
\begin{alignat}{1}
  \rho_t + \nabla\cdot(\rho\vec v) &= 0 \mylabel{eq:continuity}
\end{alignat}
and the momentum equations
\begin{alignat}{1}
  (\rho\vec v)_t+\nabla\cdot(\rho\vec v\otimes\vec v)+\nabla p &= 0
\end{alignat}
The pressure $p$ is a strictly increasing smooth function of $\rho$. The \defm{sound speed} $c$ is 
$$c=\sqrt{\frac{dp}{d\rho}(\rho)}.$$

If we assume irrotationality
$$\nabla\times\vec v,$$
then we may take
$$\vec v=\nabla\phi$$
for a scalar potential $\phi$. Assuming smooth flow, the momentum equations yield
\begin{alignat}{1}
  \rho &= \pi^{-1}(A-\phi_t-\frac12|\nabla\phi|^2) \myeqlabel{eq:rhofun}
\end{alignat}
where $A$ is a global constant and where 
\begin{alignat}{1}
    \frac{d\pi}{d\rho}=\frac{1}{\rho}\cdot\frac{dp}{d\rho}=\rho^{-1}c^2. \myeqlabel{eq:pideriv}
\end{alignat}
The remaining continuity equation \myeqref{eq:continuity} is \defm{unsteady potential flow}.

For any $t\neq0$ we may change from standard coordinates $(t,x,y)$ to \defm{similarity coordinates} $(t,\xi,\eta)$
with $\vec\xi=(\xi,\eta)=(x/t,y/t)$. A flow is \defm{self-similar} if $\rho,\vec v$ are functions of $\xi,\eta$ alone, without
explicit dependence on $t$. In potential flow that corresponds to the ansatz
$$\phi(t,x,y)=t\psi(x/t,y/t).$$
By differentiating the divergence form \eqref{eq:continuity} of potential flow and using \myeqref{eq:rhofun} and \myeqref{eq:pideriv}, 
we obtain the form
\begin{alignat}{1}
  (c^2I-(\nabla\psi-\vec\xi)^2):\nabla^2\psi &= 0. \myeqlabel{eq:nondivpsi}
\end{alignat}
Here $A:B$ is the Frobenius product $\trace(A^TB)$, $\vec w^2:=\vec w\otimes\vec w=\vec w\vec w^T$ (not $\vec w^T\vec w$) 
and $\nabla^2$ is accordingly the Hessian. In coordinates:
\begin{alignat*}{1}
  (c^2-(\psi_\xi-\xi)^2)\psi_{\xi\xi}-2(\psi_\xi-\xi)(\psi_\eta-\eta)\psi_{\xi\eta}
  +(c^2-(\psi_\eta-\eta)^2)\psi_{\eta\eta} &= 0.
\end{alignat*}
It is sometimes more convenient to use the \emph{pseudo-potential}
$$\chi:=\psi-\frac12|\vec\xi|^2$$
which yields
\begin{alignat}{1}
  (c^2I-\nabla\chi^2):\nabla^2\chi + 2c^2 - |\nabla\chi|^2 = 0. \myeqlabel{eq:nondivchi}
\end{alignat}
We choose $A=0$ so that
\begin{alignat}{1}
    \rho &= \pi^{-1}\big(-\chi-\frac12|\nabla\chi|^2\big). \myeqlabel{eq:pi-ss}
\end{alignat}
\myeqref{eq:nondivchi} is manifestly translation-invariant. Translation is nontrivial: in $(t,x,y)$ 
coordinates it corresponds
to a change of inertial frame
$$\vec v\leftarrow\vec v-\vec w,\qquad \vec\xi=\vec x/t \leftarrow\vec\xi-\vec w,$$
where $\vec w$ is the velocity of the new frame relative to the old one. Obviously the \defm{pseudo-velocity}
$$\vec z:=\nabla\chi=\nabla\psi-\vec\xi$$
does not change.

Self-similar potential flow is mixed-type; the local type is determined by the coefficient matrix $c^2I-\nabla\chi^2$
which is positive definite if and only if $L<1$, where 
$$L:=\frac{|\vec z|}{c}=\frac{|\vec v-\vec x/t|}{c}$$
is called \defm{pseudo-Mach number}; for $L>1$ the equation is hyperbolic.

\subsection{Shock conditions}

The weak solutions of potential flow are defined by the divergence-form continuity equation \eqref{eq:continuity}.
Its self-similar form is
$$\nabla\cdot(\rho\nabla\chi)+2\rho = 0.$$
The corresponding Rankine-Hugoniot condition is 
\begin{alignat}{1}
  \rho_uz^n_u &= \rho_dz^n_d \myeqlabel{eq:rh-z} 
\end{alignat}
where $u,d$ indicate the limits on the \defm{upstream} and \defm{downstream} side and $z^n$, $z^t$ are the normal and tangential
component of $\vec z$.
As the equation is second-order, we must additionally require continuity of the potential:
\begin{alignat}{1}
  \psi^u &= \psi^d.
\end{alignat}
By taking a tangential derivative, we obtain
\begin{alignat}{1}
  z^t_u &= z^t_d =: z^t.
\end{alignat}

Observing that $\sigma=\vec\xi\cdot\vec n$ is the shock speed, we obtain the more familiar form
\begin{alignat}{1}
  \rho_uv^n_u - \rho_dv^n_d &= \sigma(\rho_u-\rho_d), \myeqlabel{eq:rh-v} \\
  v^t_u &= v^t_d =: v^t. \myeqlabel{eq:vtan}
\end{alignat}

Fix the unit shock normal $\vec n$ so that $z^n_u>0$ which implies $z^n_d>0$ as well.
To avoid expansion shocks we must require the admissibility condition $z^n_u\geq z^n_d$, which is equivalent to
\begin{alignat}{1}
  v^n_u &\geq v^n_d.
\end{alignat}
We chose the unit tangent $\vec t$ to be $90^\circ$ counterclockwise from $\vec n$.

By \myeqref{eq:vtan} the tangential components of the velocity are continuous across the shock,
so the velocity jump is normal. 
Assuming $v^n_u>v^n_d$ (positive shock strength), we can express the shock normal as 
\begin{alignat}{1}
  \vec n &= \frac{\vec v_u-\vec v_d}{|\vec v_u-\vec v_d|}. \myeqlabel{eq:normal-v}
\end{alignat}

\subsection{Shock polar}

In our problem the upstream regions are constant and determined. Let $\psi$ be the potential in the downstream
region, $\psi^I$ the potential upstream (ditto for $\chi$, $\rho$, ...).
We substitute \myeqref{eq:normal-v} into \myeqref{eq:rh-z} to obtain the shock condition
\begin{alignat}{1}
  g(\nabla\psi,\psi,\vec\xi) := \big(\pi^{-1}(-\chi-\frac12|\nabla\chi|^2)\nabla\chi-\rho^I\nabla\chi^I\big)
  \cdot\frac{\nabla\psi^I-\nabla\psi}{|\nabla\psi^I-\nabla\psi|} = 0. \myeqlabel{eq:g}
\end{alignat}
The shock polar (see Figure \myref{fig:spolar-right}) is the curve of $\vec v_d$ that we obtain
when holding the shock in a fixed $\vec\xi$ and keeping the upstream constant while varying the normal. 
For a fixed $\vec\xi$, $\nabla\chi^I$ is fixed and $\psi=\psi(\vec\xi)=\psi^I(\vec\xi)$ is fixed as well.
Having eliminated the normal in \myeqref{eq:g}, we see that the shock 
polar is the curve of solutions $\vec v=\nabla\psi$
of $g(\vec v,\psi,\vec\xi)=0$. Hence the vector
$$g_{\vec v}=(\frac{\partial g}{\partial v_1},\frac{\partial g}{\partial v_2})$$
is normal to the shock polar, by the implicit function theorem.
Omitting a positive scalar factor, it is given by the explicit formula
\begin{alignat}{1}
  g_{\vec v} &\sim (1-(z^n_d/c)^2)\vec n-z^t(\frac{1}{z^n_u}+c^{-2}z^n_d)\vec t, \myeqlabel{eq:gv}
\end{alignat}
as we show in \myeqref{eq:g-vecv}.

For transonic shocks, which are our focus,
the downstream is elliptic, i.e.\ $1>L_d=|\vec z_d|/c\geq z^n_d/c$. In this case the coefficient of
$\vec n$ in \myeqref{eq:gv} is necessarily positive, so $g_{\vec v}\neq 0$. 

In Figure \myref{fig:spolar-right} right the leftmost point of the polar is a \emph{pseudo-normal} 
shock: $z^t=0$. In this case
$g_{\vec v}$ points in the same direction as $\vec n$, hence right.
Therefore $g_{\vec v}$ is an \emph{inner} normal\footnote{not necessarily unit} 
to the admissible part of the shock polar. 

In local RR the reflected shock must yield $\vec v_3$ parallel to the wall. 
In Figure \myref{fig:spolar-right} right, $\vec v_d$ for the weak shock 
(base in origin, tip in W) forms a blunt
angle with inner normals of the shock polar whereas $\vec v_d$ for the strong shock (tip in $K$) 
forms a sharp angle. 
For the critical angle there is a single shock which is a limit of the weak and strong sides, so the 
angle is right (see $\tau_*$ in Figure \myref{fig:spolar-right} right). 
This motivates the following definition:
\begin{definition}
    \label{def:type}%
  A shock is called \defm{weak-type} (in a particular point $\vec\xi$ in self-similar coordinates) if 
  \begin{alignat}{1}
    g_{\vec v}\cdot\vec z_d &< 0, \myeqlabel{eq:weaktype}
  \end{alignat}
  \defm{strong-type} if negative, \defm{critical-type} if zero. 
\end{definition}
The definition has three pleasant properties: it coincides with the standard definition in the case 
of strictly convex polars,
it generalizes the definition of weak/strong-type to non-convex cases\footnote{In such cases, there
may be three or more reflected shocks that yield $\vec v_3$ tangential to the wall.}, and finally 
the sign condition is precisely what is needed for elliptic corner regularity.

\subsection{Polytropic pressure}

Throughout the paper we consider only the standard polytropic pressure law:
$$p(\rho)=\frac{c_0^2\rho_0}{\gamma}\big(\frac{\rho}{\rho_0}\big)^\gamma$$
with $\gamma\in(1,\infty)$, where $c_0,\rho_0$ are constants. With this choice,
$$c^2=c_0^2\big(\frac{\rho}{\rho_0}\big)^{\gamma-1}.$$

\begin{theorem}
  \label{th:shockpolar}%
    Consider arbitrary $c_u,\rho_u>0$ and $M_u\in(1,\infty)$ and set $\vec v_u=(M_uc_u,0)$.
    For each $\beta\in(-90^\circ,90^\circ)$ there is a steady shock 
    with downstream unit normal $\vec n=(\cos\beta,\sin\beta)$.
    Its downstream state $\rho_d,c_d,\vec v_d$ depends smoothly on $\beta$. 
    Let $\tau$ be counterclockwise angle from
    $\vec v_u$ to $\vec v_d$.
    We restrict 
    $|\beta|<\arccos\frac{1}{M_u}$
    so that the shock is admissible.

    Then the \defm{shock polar} $\beta\mapsto\vec v_d$ is smooth and strictly convex, 
    with $\partial_\beta\vec v_d$ 
    nowhere zero. 

    There is an angle $\tau_*\in(0^\circ,90^\circ)$ so that each $\tau\in(-\tau_*,\tau_*)$
    is attained for two different $\beta$.
    The one with smaller $|\vec v_d|$ yields a strong-type shock, the other one weak-type.
    For $|\tau|=\tau_*$ they are identical and critical-type.

    There is a $\tau_s\in(0,\tau_*)$ so that the weak-type shocks are supersonic
    for $|\tau|>\tau_s$, transonic for $|\tau|<\tau_s$. The other types are always transonic.
\end{theorem}
\begin{proof}
    We refer to \cite{\pmpaper}, especially Proposition \pmref{prop:shockpolar}, which establishes
    existence and smooth dependence of admissible shocks. By \pmc{\pmeqref{eq:DvdyDnva}}
    $\partial_\beta\vec v\neq 0$ at all $\beta$. 

    As shown earlier, $g_{\vec v}$ in \myeqref{eq:gv} is an inner normal to the shock polar everywhere.
    Multiply it with a positive factor to obtain $q=\vec n-A\vec t$
    where
    $$A=\frac{v^t(1/v^n_u+M^n_d/c_d)}{1-(M^n_d)^2}.$$
    $A$ is decreasing in $\beta\leq0$, because by \pmc{Proposition \pmref{prop:shockpolar}}
    $c_d>0$ is increasing, $v^n_u>0$ is increasing, $M^n_d>0$ is decreasing, 
    $v^t>0$ is decreasing. Hence $\partial_\beta A\leq0$. 
    Moreover
    $$\partial_\beta q = A\vec n+(1-(\partial_\beta A))\vec t,$$
    ($\vec t$ is counterclockwise from $\vec n$), so
    $$q\times\partial_\beta q = 1-\partial_\beta A+A^2>0.$$
    This implies that the upper half of the shock polar is strictly convex. 
    By vertical symmetry and smoothness the entire polar is strictly convex.

    The shock polar is compact when adding the ``vanishing'' shock $\vec v_d=\vec v_u$.
    Moreover $|\tau|<90^\circ$, so there is a maximum $\tau_*\in(0^\circ,90^\circ)$. 
    By convexity there are exactly two points
    on the polar for $|\tau|<\tau_*$, which are the intersections of the line of multiples of $\vec v_d$
    with the polar. 
    As $g_{\vec v}$ is an inner normal, necessarily $g_{\vec v}\cdot\vec v_d>0$ for the point closer
    to the origin (strong-type), with opposite sign for the other (weak-type).
    
    $$g_{\vec v}\cdot\vec v=(1-(v^n_d/c_d)^2)v^n_d-(v^t)^2(1/v^n_u+c_d^{-2}v^n_d)
    =v^n_d\Big((1-M_d^2)-\frac{(v^t)^2}{v^n_dv^n_u}\Big).$$
    If $M_d\geq 1$, then the right-hand side is negative, so the shock is weak-type.
    By Proposition \pmref{prop:shockpolar}, $M_d$ is strictly decreasing in $|\beta|$, 
    so there is a unique $\tau_s$ so that the weak-type shock is transonic for $|\tau|>\tau_s$,
    supersonic for $|\tau|<\tau_s$.
\end{proof}

\section{Perturbations of weak trivial RR}

\subsection{Coordinate transform}

We consider a trivial RR as in Figure \myref{fig:bigger90} center or Figure \myref{fig:index} left.
All lines and curves exclude endpoints by default. 
We use the following notation (see Figure \ref{fig:index} left):
Let $\hat B$ be the reflection wall, $\hat A$ the opposite wall, $W$ the open convex cone enclosed by them.
Let $\vec n_A,\vec n_B$ be the outer (with respect to $W$) unit normals of $\hat A,\hat B$.
Let the origin the the corner between $A,B$.
Let $S$ the reflected shock, $\vec\xi_A=(\xi_A,0)$, $\vec\xi_B=(\xi_B,\eta_B)$ (note $\xi_A=\xi_B$) the points 
where it meets $\hat A$ resp.\ $\hat B$. Let $A,B$ be the segments of $\hat A,\hat B$ 
from the corner $(0,0)$ to 
$\vec\xi_A,\vec\xi_B$; let $\Omega$ be the triangle enclosed by $A,B,S$.

The velocity in $\Omega$ is zero in the chosen coordinates,
so the velocity potential $\psi$ is constant $=\psi^0$ in $\Omega$.
Let $\vec v_I=(v^x_I,0)$ be the $2$-sector, $\psi^I$ the corresponding potential.

In self-similar flow, $\Omega$ is a uniformly elliptic region whereas the rest of $W$ is uniformly hyperbolic. 

The shock is a free boundary. To linearize the problem, we first devise a transform from
$\vec\xi=(\xi,\eta)$ to fixed coordinates $\vec\sigma=(\sigma,\zeta$).

Given a function $\psi\in C^2(\Omega)\cap C^1(\overline\Omega)$.
Consider a ray starting in the origin and passing through $(\xi_A,\zeta)\in S$. 
$\psi^I$ is strictly monotone along any such ray, so there is a unique point
$\vx$ with 
\begin{alignat}{1}
    \psi^I(\vx) &= \psi(\xi_A,\zeta). \myeqlabel{eq:psimatch}
\end{alignat}
$(\sigma,\zeta)\in\Omega$ is mapped to $(\frac{\xi\xi_A}{\sigma},\frac{\eta\xi_A}{\sigma})$.

This coordinate transform allows to state our problem in a \emph{fixed} domain $\Omega$.
By \myeqref{eq:psimatch}, $\psi$ mapped to $\vec\xi$ coordinates
satisfies the first shock condition, $\psi=\psi^I$, automatically.
Then \myeqref{eq:g} can be used as the second shock condition.

\subsection{Linearization}

\mylabel{section:linearization}

We regard our problem as an operator equation
$$F(\psi)=0$$ 
where $F:X\rightarrow Y$, $X,Y$ Banach spaces with $X\subset C^2(\Omega)\cap C^1(\overline\Omega)$.
The map is the composition of three steps: first $\psi\in X$ is transformed from $(\sigma,\zeta)$
to $\vx$ coordinates, then mapped to the tuple
\begin{alignat}{2}
  & \Big((c^2I-\nabla\chi^2):\nabla^2\psi, &\qquad\text{[Interior]} \notag \\
  & \nabla\psi_{|A}\cdot\vec n_A, &\qquad\text{[Slip condition at $A$]} \notag \\
  & \nabla\psi_{|B}\cdot\vec n_B, &\qquad\text{[Slip condition at $B$]} \notag \\
  & g(\nabla\psi,\psi,\vec\xi)\Big). &\qquad\text{[Shock condition]}
  \myeqlabel{eq:nlmap}
\end{alignat}
Finally, pull back this tuple to $(\sigma,\zeta)$ coordinates. 

$X$ will be specified later since we have to consider an entire scale of such spaces.
$F$ will be a nonlinear $C^1$ map from $X$ to $Y$.
We intend to apply the implicit function theorem. To this end we need to study
the Fr\'echet derivative $F'(\psi^0)$ of $F$ with respect to $\psi$ at $\psi=\psi^0$.

The derivative is computed by considering a first variation $\psi'\in X$ of $\psi^0$.
We consider $\psi(\vec\sigma)=\psi^0+t\psi'(\vec\sigma)$ 
and compute the derivative $\frac{d}{dt}$ of $F(\psi)$
and other expressions, evaluated at $t=0$. This derivative will be written $\vec\xi'$, $\rho'$, 
$F(\psi)'$, etc. Obviously the usual calculus rules apply. 

The following calculations are quite similar to \pmc{Proposition 4.14.3}. 
The results are simplified by two facts: 
$\psi=\psi^0$ yields an identity $(\xi,\eta)=(\sigma,\zeta)$, and 
$\nabla\psi^0=\nabla^2\psi^0=0$. 

All derivatives are evaluated at $\psi=\psi^0$; we omit arguments where they are clear from the context.

\begin{alignat}{1}
  (\nabla_{\vec\xi}\psi)' 
  &= \big(\nabla_{\vec\xi}\vec\sigma^T\nabla_{\vec\sigma}\psi\big)' \notag\\
  &= \big(\nabla_{\vec\xi}\vec\sigma^T\big)'\subeq{\nabla_{\vec\sigma}\psi}{=0}
  +\nabla_{\vec\xi}\vec\sigma^T(\nabla_{\vec\sigma}\psi)' 
  =\nabla_{\vec\xi}\vec\sigma^T\nabla_{\vec\sigma}\psi' = \nabla_{\vec\xi}\psi'.\notag
\end{alignat}
\begin{alignat}{1}
  (\nabla^2_{\vec\xi}\psi)'
  &= \Big(\sum_k\frac{\partial\psi}{\partial\sigma^k}\nabla_{\vec\xi}^2\sigma^k
  +\nabla_{\vec\xi}\vec\sigma^T\nabla_{\vec\sigma}^2\psi\nabla_{\vec\xi}^T\vec\sigma\Big)' \notag\\
  &= \Big(\sum_k\frac{\partial\psi}{\partial\sigma^k}\Big)'\nabla_{\vec\xi}^2\sigma^k
  + \sum_k\subeq{\frac{\partial\psi}{\partial\sigma^k}}{=0}\big(\nabla_{\vec\xi}^2\sigma^k\big)' \notag\\
  &+ \big(\nabla_{\vec\xi}\vec\sigma^T\big)'\subeq{\nabla_{\vec\sigma}^2\psi}{=0}\nabla_{\vec\xi}^T\vec\sigma
  +\nabla_{\vec\xi}\vec\sigma^T\big(\nabla_{\vec\sigma}^2\psi\big)'\nabla_{\vec\xi}^T\vec\sigma
  +\nabla_{\vec\xi}\vec\sigma^T\subeq{\nabla_{\vec\sigma}^2\psi}{=0}(\nabla_{\vec\xi}^T\vec\sigma)' \notag\\
  &= \sum_k\frac{\partial\psi'}{\partial\sigma^k}\nabla_{\vec\xi}^2\sigma^k
  +\nabla_{\vec\xi}\vec\sigma^T\nabla_{\vec\sigma}^2\psi'\nabla_{\vec\xi}^T\vec\sigma = \nabla_{\vec\xi}^2\psi'\notag
\end{alignat}

Fr\'echet derivative of the interior equation:
\begin{alignat}{1}
  0 &= \big(c^2I-(\nabla_{\vec\xi}\chi)^2\big)':\subeq{\nabla_{\vec\xi}^2\psi}{=0} + 
  \big(c^2I-(\nabla_{\vec\xi}\chi)^2\big):(\nabla_{\vec\xi}^2\psi)' \notag\\
  &= \big(c^2I-(\nabla_{\vec\xi}\chi)^2\big):\nabla_{\vec\xi}^2\psi'
  = \big(c^2I-\vec\xi^2\big):\nabla_{\vec\xi}^2\psi'
  \myeqlabel{eq:pde-lin} 
\end{alignat}
The resulting right-hand side is a linear elliptic operator without zeroth-order term,
applied to $\psi'$. The classical maximum principle shows that $\psi'$ cannot have a minimum in the interior.

The wall conditions linearize to
\begin{alignat}{1}
  \nabla\psi'\cdot\vec n &= 0. \myeqlabel{eq:wall-lin}
\end{alignat}

For the shock condition, we consider \myeqref{eq:g}.
First, hold $\vx,\psi$ fixed and very $\nabla_{\vec\xi}\psi$.
The variation of the normal expression \myeqref{eq:normal-v} is
\begin{alignat}{1}
    (\frac{\vec v_I-\nabla_{\vec\xi}\psi}{|\vec v_I-\nabla_{\vec\xi}\psi|})'
    &= \frac{-1}{|\vec v_I-\nabla_{\vec\xi}\psi|}
    \subeq{\left(1-\big(\subeq{\frac{\vec v_I-\nabla_{\vec\xi}\psi}{|\vec v_I
                    -\nabla_{\vec\xi}\psi|}}{=\vec n}\big)^2\right)}{=(\vec t)^2}\nabla_{\vec\xi}\psi'
    =\frac{-(\psi')_t}{|\vec v_I-\nabla_{\vec\xi}\psi|}\vec t \myeqlabel{eq:fdnor}
\end{alignat}
Moreover, 
\begin{alignat*}{1}
    (\rho\nabla_{\vx}\chi-\rho^I\nabla_{\vx}\chi^I)' 
    &\overset{\text{\myeqref{eq:pi-ss}}}= \Big(\pi^{-1}(-\chi-\frac12|\nabla_{\vx}\chi|^2)\nabla_{\vx}\chi\Big)' \\
    &\overset{\text{\myeqref{eq:pideriv}}}= \rho(I-c^{-2}\nabla_{\vx}\chi^2)\nabla\psi'.
\end{alignat*}
Both combined, we use the shock relations $\rho_I=\rho\chi_n/\chi^I_n$ and $\chi_t=\chi^I_t$ to compute
\begin{alignat}{1}
    (g(\nabla_{\vx}\psi,\psi,\vx))' 
    &= \Big((\rho\nabla_{\vec\xi}\chi-\rho^I\nabla_{\vec\xi}\chi^I)\cdot
    \frac{\vec v_I-\nabla_{\vec\xi}\psi}{|\vec v_I-\nabla_{\vec\xi}\psi|}\Big)' \notag\\
    &= 
    \Big((\rho\nabla_{\vec\xi}\chi-\rho^I\nabla_{\vec\xi}\chi^I)\Big)'\cdot\vec n
    +
    (\rho\nabla_{\vec\xi}\chi-\rho^I\nabla_{\vec\xi}\chi^I)\cdot
    \Big(\frac{\vec v_I-\nabla_{\vec\xi}\psi}{|\vec v_I-\nabla_{\vec\xi}\psi|}\Big)' \notag\\
    &= \rho\vec n^T(I-c^{-2}\nabla_{\vx}\chi^2)\nabla_{\vx}\psi'
    -\frac{\rho\chi_t-\rho^I\chi^I_t}{|\vec v_I-\nabla_{\vx}\psi|}\psi'_t \notag\\
    &= \subeq{\rho\Big((1-c^{-2}\chi_n^2)\vec n
        -\chi_t\big(\frac{1}{\chi^I_n}+c^{-2}\chi_n\big)\vec t\Big)}{=:g_{\vec v}}\cdot\nabla_{\vx}\psi'.
    \myeqlabel{eq:g-vecv}
\end{alignat}

Now we hold $\nabla_{\vx}\psi$ fixed and vary $\psi$. 
$\psi=\psi^I=\psi^I(0,0)+v^x_I\xi$ on the shock, so we can use
\begin{alignat}{1}
    \xi' &= (v^x_I)^{-1}\psi'. \myeqlabel{eq:fdeta}
\end{alignat}
Moreover, the variation of the ``normal'' $\vec v_I-\nabla_{\xi}\psi$ is zero here, so:
\begin{alignat*}{1}
    (g)'
    &= \Big(\pi^{-1}\big(-\psi+\frac12|\vx|^2-\frac12|\nabla_{\vx}\psi-\vx|^2\big)(\nabla_{\vx}\psi-\vx)
    -\rho^I(\vec v_I-\vx)\Big)'
    \cdot\subeq{\vec n}{=(1,0)} \notag \\
    &= \big(\rho c^{-2}\cdot(-\psi'+\subeq{\vx\cdot(\vx')+(\nabla_{\vx}\psi-\vx)\cdot(\vx)'}{=\nabla_{\vx}\psi\cdot(\vx)'=0})
    \nabla_{\vx}\chi-\rho(\vx)'+\rho_I(\vx)'\big)\cdot\subeq{\vec n}{=(1,0)} \notag \\
    &= -\rho c^{-2}\chi_n\psi'+(\rho_I-\rho)\xi' \notag \overset{\myeqref{eq:fdeta}}{=} -\rho c^{-2}\chi_n\psi'+\frac{\rho_I-\rho}{v^x_I}\psi' \notag \\
    &= -\rho(\frac{1}{\chi^I_n}+c^{-2}\chi_n)\psi'
\end{alignat*}
It turns out that $\eta'$ does not appear in the final form, so the details
of the coordinate transform do not matter at all!

Altogether, when varying $\nabla\psi$ and $\psi$ at the same time, the shock relations
linearize to
\begin{alignat}{1}
    (g)' &= g_{\vec v}\cdot\nabla_{\vx}\psi'-\rho(\frac{1}{\chi^I_n}+c^{-2}\chi_n)\psi'. \myeqlabel{eq:shockl}
\end{alignat}

\subsection{Kernel}

\begin{proposition}
  \mylabel{prop:my-kernel}%
  For any $X\subset C^2(\Omega)\cap C^1(\overline\Omega)$,
  $$\dim\ker F'(\psi^0) \leq 1.$$
  If $=1$, then it is spanned by a function $\psi'$ that satisfies
  $$\psi'(\vec\xi_B)\neq 0.$$
\end{proposition}
\begin{proof}
    Assume the kernel is nontrivial.
    Let $\psi'$ be a nonzero element.

  Consider a positive local maximum (with respect to $\overline\Omega$)
  of $\psi'$ at $S\cup\{\vec\xi_A\}$. 
  A maximum at $S$ requires $\psi'_t=0$; for a maximum in $\vec\xi_A$
  this is already implied by the boundary condition $\psi'_n=0$ on $A$,
  by $C^1$ regularity in the corner, 
  since $A$ and $S$ meet at a right angle.
  The coefficients of $-\psi'_n$
  and $\psi'$ in \myeqref{eq:shockl}, the linearization of the 
  shock condition,
  have opposite sign. 
  Therefore $\psi'>0$ in the maximum point implies
  $-\psi_n>0$ which is incompatible with a local maximum
  ($\vec n$, the downstream normal, is an inner normal for $\Omega$).
  By the same argument a local negative minimum is ruled out.

  This implies in particular that $\psi'$ cannot be constant.

  $\psi'$ satisfies \myeqref{eq:pde-lin}, the linearization of the interior
  PDE, so by the classical strong maximum
  principle $\psi'$ cannot have a local extremum in $\Omega$ 
  unless it is constant.
  By the Hopf lemma, the wall boundary condition \eqref{eq:wall-lin} does not allow a local extremum
  at $A$ or $B$ unless $\psi'$ is constant.

  Assume $\psi'$ has a global maximum in $0$ (wall-wall corner). Let $B_\epsilon(0)$ be the 
  ball with radius $\epsilon$
  centered in $0$ and abbreviate $U:=B_\epsilon(0)\cap\Omega$, $I:=\partial B_\epsilon(0)\cap\Omega$.
  For sufficiently small $\epsilon>0$, $\overline I\subset\Omega\cup A\cup B$, so 
  as shown above $\psi'$ cannot attain a maximum on $\overline I$.
  Therefore $\psi'(0)>\max_{\overline I}\psi'$. 
  
  $\hat\psi:=\psi'-\psi'(0)+\delta\xi$ is a supersolution for $\delta\geq 0$:
  $$(I-c^{-2}\vec\xi^2):\nabla^2\hat\psi=(I-c^{-2}\vec\xi^2):\nabla^2\psi'=0$$
  by linearity, 
  $\hat\psi_n=0$ on $A$ and $\hat\psi_n=(\delta\xi)_n>0$ on\footnote{note $\theta>90^\circ$ for trivial RR} $B$.
  Therefore $\hat\psi$ does not attain extrema in $U$.
  For $\delta>0$ sufficiently small, 
  $$\max_{\overline I}\hat\psi=\max_{\overline I}\psi'-\psi'(0)+\delta\xi>0,$$
  while $\hat\psi(0)-\psi'(0)=0$, so the minimum of $\hat\psi$ over $\overline U$ is attained in $0$. 
  Therefore $\hat\psi^{}_\xi(0)\leq0$,
  hence $\psi'_\xi(0)\leq-\delta<0$. But the boundary conditions $\psi'_n=0$ on $\overline A,\overline B$ 
  combine to $\nabla\psi'(0)=0$ 
  --- contradiction. Hence $\psi'$ cannot have a global maximum in $0$; minima are ruled out analogously.

  Since $\psi'$ is nonzero, it must have a positive maximum or
  negative minimum somewhere. As we have shown that is not possible except 
  in $\vec\xi_R$.

  For any two elements of the kernel, a suitable linear combination 
  is zero in $\vec\xi_R$, hence zero everywhere. Thus the kernel
  cannot have dimension higher than $1$.
\end{proof}

\subsection{Type and Fredholm index}

\begin{proposition}
  \mylabel{prop:weak-gsp}%
  Consider the eigenvalues of the operator pencil for $F'(\psi^0)$ in the reflection corner $\vec\xi_B$ 
  (see Section \myref{section:operatorpencils}). 
  There is an eigenvalue $\lambda_0=\alpha_0+i\beta_0$ of multiplicity $1$ with least nonnegative $\beta_0$, and
  \begin{alignat}{1}
      \beta_0 &\begin{cases} 
          \in(0,1), & \text{if the shock is strong-type in $\vec\xi_B$,} \\
          =1, & \text{for critical-type,} \\
          >1, & \text{for weak-type.}
      \end{cases} \myeqlabel{eq:alpha0}
  \end{alignat}
\end{proposition}
\begin{proof}
    The operator $F'(\psi_0)$, with coefficients frozen in $\vec\xi_B$, 
    consists of the interior operator
    $(I-c^{-2}\vec\xi_B^2):\nabla^2\psi'$ 
    and the boundary operators $\nabla\psi'\cdot\vec n_B$ and $g_{\vec v}\cdot\nabla\psi'$.
    We choose a linear coordinate transform so that the interior operator is mapped into $\Delta\psi'$.
    This transform is a dilation in the $B$ direction. 

    Consider polar coordinates $(r,\phi)$ centered in $\vec\xi_B$.
    Let $\Gamma_2=B$, $\Gamma_1=S$, then the boundary operators take the form \myeqref{eq:beta}
    with $\gamma_2=90^\circ$ (Neumann) and (see Figure \myref{fig:corner})
    \begin{alignat}{1}
        \gamma_1 & \begin{cases}
            \in(90^\circ,\phi_2-\phi_1+90^\circ), & \text{for strong-type,} \\
            =\phi_2-\phi_1+90^\circ, & \text{for critical-type,} \\
            \in(\phi_2-\phi_1+90^\circ,180^\circ), & \text{for weak-type.}
        \end{cases}
    \end{alignat}
    To see this, note that $\nabla\chi^0=\nabla\psi^0-\vec\xi=-\vec\xi\parallel B$ on $B$.
    For a weak-type shock (Definition \myref{def:type}), $g_{\vec v}\cdot\nabla\chi^0<0$,
    so $\vec n_B\times g_{\vec v}>0$. This property is preserved under dilation along $B$,
    so $\gamma_1>\phi_2-\phi_1+90^\circ$ (see Figure \myref{fig:corner} left). The other types are analogous.

    Now \myeqref{eq:alpha} immediately implies \myeqref{eq:alpha0}.
\end{proof}

\begin{figure}
\input{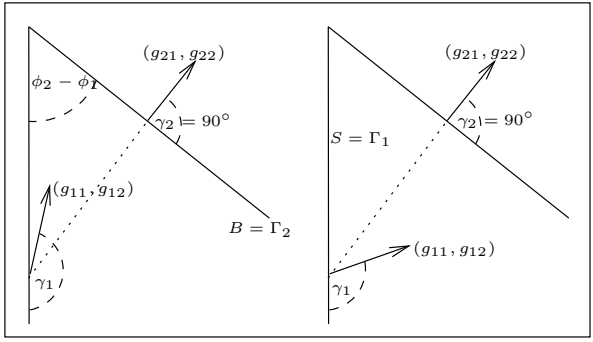}
\caption{Corner limits of the top-order parts $(g_{k1},g_{k2})\cdot\nabla\psi'$ of
the boundary operators. Left: weak-type shock. Right: strong-type shock: solutions need not be $C^1$
in the corner. A critical-type shock has $(g_{11},g_{12})$ exactly perpendicular to $B$.}
\mylabel{fig:corner}
\end{figure}

\begin{proposition}
    \mylabel{prop:wallwall-beta}%
    Consider the eigenvalues $\lambda=\alpha+i\beta$ of the 
    operator pencil of $F'(\psi^0)$ in the $A,B$ and $A,S$ corner.
    The eigenvalue with least nonnegative $\beta$ is $\lambda_0=0$. 
    The eigenvalue with next lowest nonnegative $\beta$ is $\beta_1=1/(1-\theta/180^\circ)>1$
    in the $A,B$ corner and $\beta_1=2$ in the $A,S$ corner; their multiplicity is $1$.
\end{proposition}
\begin{proof}   
    In the $A,B$ corner the interior operator is $\Delta$, with Neumann boundary operators $\partial_n$
    (so $\gamma_1=\gamma_2=90^\circ$ in the notation of Section \myref{section:operatorpencils}),
    so the calculation is straightforward. Take $\Gamma_1:=B$, $\Gamma_2:=A$. 
    $\phi_1=\theta$, $\phi_2=\pi$, then by \myref{eq:alpha}
    $$\beta_0=0, \qquad \beta_1 = \frac{\pi}{\pi-\theta}.$$

    In the $A,S$ corner the slip condition on $A$ yields $\chi_\eta=0$, so the
    interior operator is $(1-\chi_\xi^2)\partial_{\xi\xi}+\partial_{\eta\eta}$ which
    becomes $\Delta$ by dilation. 
    Moreover by \eqref{eq:shockl} both boundary operators are $\partial_n$, which are not changed by dilation.
    Hence $\beta_1=2$ by \myeqref{eq:alpha}.
\end{proof}

\begin{figure}
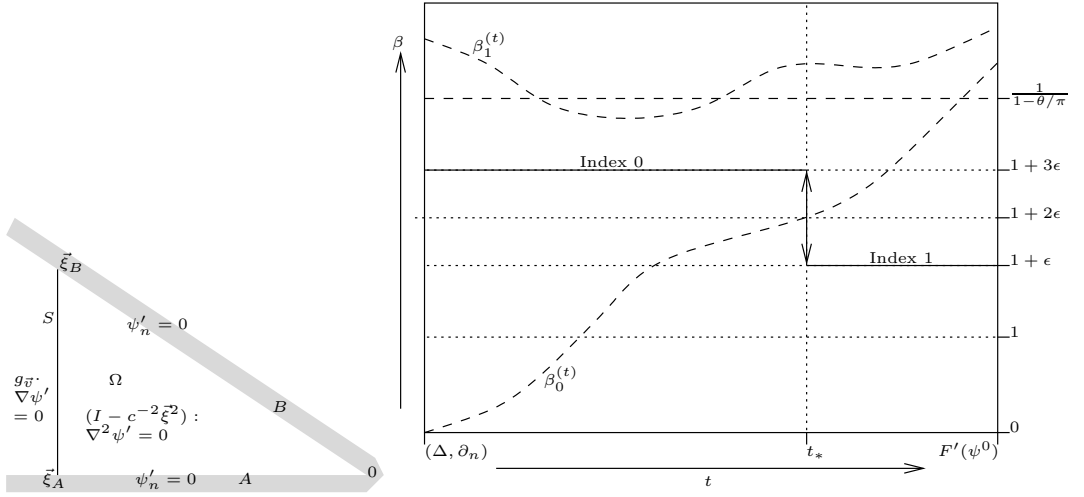

\input{elliptic.pstex_t}
\input{index.pstex_t}
\caption{Left: notation and linearized operator.
    Right: each point represents an operator $T^{(t)}:X^s_\beta\rightarrow Y^s_\beta$ ($t$ horizontal axis, $\beta$
    vertical axis). The operators are Fredholm except on the dashed curves, where some corner operator pencils
    has an eigenvalue. Between the curve the index is constant. Across the curves the index jumps by 
    the eigenvalue multiplicity.}
\mylabel{fig:index}
\end{figure}

\begin{proposition}
  \mylabel{prop:my-index}%
  Consider a weak-type trivial reflection. 
  Let $s\in(2,3)$. 
  For $\epsilon>0$ sufficiently small, the Fredholm index of $F'(\psi^0)$ 
  as a map from $X^s_{1+\epsilon}$ to $Y^s_{1+\epsilon}$ (defined in Section \myref{section:operatorpencils})
  is $1$. 
\end{proposition}
\begin{proof}
    Consider the operator $\Delta-I$ on the triangle $\Omega_0$ with Neumann boundary operators
    $\partial_n$ on $A,B,S$.
    \cite[Theorem 1.4]{lieberman-crelle-1988} yields 
    $s\in(2,3)$ and $\epsilon\in(0,s-2)$ so that the operator 
    is a linear isomorphism on $X^s_{1+3\epsilon}$ onto $Y^s_{1+3\epsilon}$. 
    The space $X^s_{1+3\epsilon}$ defined in the present paper corresponds to $H^{(-1-3\epsilon)}_s$ 
    in his notation, \emph{except} that his weights are with respect to $\partial\Omega$, not $\Sigma$.
    But $\partial\Omega-\Sigma$ consists of line segments, 
    so classical potential theory\footnote{moreover $s$ determines
        only regularity away from the corners, so it can be improved to any $s>2$}
    \cite[Lemma 6.27]{gilbarg-trudinger} extends his result to our case.

    Now we choose a 
    family of operators $t\mapsto T^{(t)}$ so that $T^{(0)}=(\Delta-I,\partial_n,\partial_n,\partial_n)$
    and $T^{(1)}=F'(\psi^0)$ (see Figure \myref{fig:index}).
    The family is chosen continuous in $t$ with respect to all operator norms we consider,
    which is easily achieved by choosing a continuous family of coefficients for interior and boundary operators.
    We choose the family so that $\beta_0,\beta_1$ in the $A,B$ and $A,S$ corners are constant in $t$ 
    (Proposition \myref{prop:wallwall-beta}). 
    If $\beta^{(t)}_j$ ($j=0,1$) are the two lowest nonnegative imaginary parts of eigenvalues of the 
    operator pencils in the $\vec\xi_B$ corner,
    then $t\mapsto\beta^{(t)}_j$ are continuous as well. $\beta^{(0)}_0=0$ whereas
    Proposition \myref{prop:weak-gsp} shows that $\beta^{(1)}_0>1$. 
    By choosing suitable coefficient families in the reflection corner we can make 
    $\beta^{(t)}_0$ strictly increasing in $t$. Moreover $\beta^{(0)}_1,\beta^{(1)}_1>1$, so we can achieve 
    $\beta^{(t)}_1>1$.
    
    \cite[Theorem 6.3]{mazya-plamenevskii} yields\footnote{
      Their weighted H\"older spaces are homogeneous; 
      for our inhomogeneous spaces, 
      $\oplus\Pi_1$ (set of polynomials of degree $\leq 1$) is added for each corner,
      which is only a finite-dimensional change.} 
    that $T^{(t)}:X^s_{1+\epsilon}\rightarrow Y^s_{1+\epsilon}$ 
    is a Fredholm operator if $1+\epsilon$ is not the imaginary 
    part of an operator pencil eigenvalue in any corner. 

    Choose $\epsilon\in(0,1)$ sufficiently small (not larger than above)
    so that $1+3\epsilon<\beta^{(t)}_1,1/(1-\theta/180^\circ)$ for all $t$ (see Figure \myref{fig:index} right).
    Let $t_*\in(0,1)$ be such that $\beta^{(t_*)}_0=1+2\epsilon$. 
    Then $t\in[0,t_*]\mapsto T^{(t)}:X^s_{1+3\epsilon}\rightarrow Y^s_{1+3\epsilon}$ 
    and $t\in[t_*,1]\mapsto T^{(t)}:X^s_{1+\epsilon}\rightarrow Y^s_{1+\epsilon}$ are both continuous families
    of Fredholm operators. By Fredholm theory the index of each family is constant.
    $T^{(0)}:X^s_{1+\epsilon}\rightarrow Y^s_{1+\epsilon}$ is an isomorphism, as shown above, i.e.\ has index $0$.

    The interval $[1+\epsilon,1+3\epsilon]$ contains only one eigenvalue of a corner operator pencil of $T^{t_*}$,
    namely $\beta_1^{(t_*)}=1+2\epsilon$; its multiplicity is $1$. 
    Hence \cite[Theorem 6.4]{mazya-plamenevskii} shows that $\dim(Z/X_+)=1$ in Proposition 
    \myref{prop:index-pm},
    where we choose $X_-=X^s_{1+3\epsilon}$, $Y_-=Y^s_{1+3\epsilon}$, $X_+=X^s_{1+\epsilon}$, $Y_+=Y^s_{1+\epsilon}$
    and $A_\pm=T^{(t_*)}$. Therefore $\find A_--\find A_+=1$, so
    $T^{(1)}:X^s_{1+\epsilon}\rightarrow Y^s_{1+\epsilon}$ has index $1$.
\end{proof}

\begin{remark}
    The proof requires $\beta_0>1$, which is not satisfied for critical-
    or strong-type shocks in the reflection corner. This is the crucial difference 
    to weak-type shocks. Note that the value of $\beta_0$ is a purely local property;
    the chosen far-field perturbation is not significant.
\end{remark}

\subsection{Perturbation}

\begin{theorem}
    \mylabel{th:sonic-wrong}%
    Consider a weak-type trivial transonic RR, with parameters $\vec p_0=(M_1,\theta,\alpha)$.
    There is a ball $U$ of radius $r>0$ around $\vec p_0$ so that there is another global weak-type RR
    for any $\vec p\in U$.
\end{theorem}
\begin{proof}
    Since the downstream state of a shock depends smoothly on the shock normal, location and upstream state,
    the shock polar varies smoothly with $M_u$. 
    Necessarily $M_2>1$, so the incident shock is weak-type like the reflected shock.
    Therefore, sufficiently small perturbations $\vec p$ yield a new local RR which is close
    to the old one.
    In particular the perturbation of the reflection point is small and the reflected shock is still
    weak-type.

    By Proposition \myref{prop:my-index}, the Fredholm index of $F'(\psi^0)$ is $1$. 
    By Proposition \myref{prop:my-kernel} the kernel
    has dimension $1$, so the codimension of the range is $0$. Therefore we can apply the implicit
    function theorem, with a single real free parameter. 
    By Proposition \myref{prop:my-kernel} we can use $\psi(\vx_B)$ as free parameter,
    which corresponds to changing the reflection point.

    Therefore we obtain a new elliptic region for sufficiently small perturbations of the reflection point,
    while satisfying both shock conditions.
    After extending the solution to the entire domain by adding incident shock and hyperbolic regions, 
    we have obtained a global transonic weak-type RR. 
\end{proof}

\section{Corner domains}

Here we adapt some literature results to our case.
For details, see \cite{grisvard}, \cite{mazya-plamenevskii}, \cite{mazya-nazarov-plamenevskii-volume-one}, 
and \cite{nazarov-plamenevskii}.

\subsection{Weighted H\"older spaces}

Consider a bounded open simply connected Lipschitz domain $\Omega\subset\R^2$.
Let $\Gamma_k$ ($k=1,\dotsc,m$) be pairwise disjoint line segments with (excluded) endpoints $y_{k-1},y_k$
(set $y_0:=y_m,\Gamma_0:=\Gamma_m$ for simplicity). Set $\Sigma:=\{y_1,\dotsc,y_m\}$.
Assume $\partial\Omega=\bigcup_{k=1}^m\overline\Gamma_k$.
Let $\Gamma_1,\dotsc,\Gamma_m$ pass around $\Omega$ clockwise, so that $\Omega$ 
lies counterclockwise from $\Gamma_k$ to $\Gamma_{k+1}$
near each corner $y_k$.

\begin{definition}
    \label{def:weighted}%
    Let $\beta\in\R$, $s\in(0,\infty)-\Z$. 
    Abbreviate $\overline\Omega_r:=\overline\Omega-B_r(\Sigma)$ where $B_r$ is the $r$-neighbourhood.
    For $u\in C^s(\overline\Omega,\Sigma)$ we define \defm{weighted H\"older norms}
    \begin{alignat}{1}
        \|u\|_{C^s_\beta(\overline\Omega,\Sigma)}:= \limsup_{r\downarrow 0}r^{s-\beta}\|u\|_{C^s(\overline\Omega_r)}.
    \end{alignat}
    Then $C^s_\beta(\overline\Omega,\Sigma)$ is the set of $u$ with finite norm. The definitions for $\Gamma_j$ in place of $\Omega$
    are analogous.
\end{definition}

Non-integer $\beta$ corresponds to the lowest exponent of $r^\beta$ behaviour allowed in a corner;
note that $C^s_\beta(\overline\Omega,\Sigma)\subset C^\beta(\overline\Omega)$.

\subsection{Operator pencils}
\mylabel{section:operatorpencils}

Consider the operator of a linear second-order elliptic boundary value problem:
\begin{alignat}{1}
  L(x)u &:= \sum_{i,j=1}^2a_{ij}(x)\frac{\partial^2}{\partial x_i\partial x_j}u+\sum_{i=1}^2b_i(x)\frac{\partial}{\partial x_i}u+c(x)u 
  \qquad\text{in $\Omega$,} \\
  B_k(x)u &:= \sum_{i=1}^2g_{ki}(x)\frac{\partial}{\partial x_i}u+h_k(x)u\qquad\text{on $\Gamma_k$, $k=1,\dotsc,m$.}
\end{alignat}
We assume that the coefficients $a_{ij},b_i,c$ are smooth on $\overline\Omega$
and $g_{ki},h_k$ smooth on $\overline\Gamma_k$. 
We write $B=(B_1,\dotsc,B_m)$. Let $s\in(2,\infty)-\Z$; we use the convenient abbreviations
$$
X^s_\beta:=C^s_\beta(\overline\Omega,\Sigma), \qquad 
Y^s_\beta:=C^{s-2}_{\beta-2}(\overline\Omega,\Sigma)\times\prod_{k=1}^mC^{s-1}_{\beta-1}(\overline\Gamma_i,\Sigma).
$$
$(L,B):X^s_\beta\rightarrow Y^s_\beta$ is a continuous linear operator.

Whenever $L(y_1)$ is elliptic, we can find a linear invertible coordinate transformation so that the 
leading-order part of $L(y_1)$ transforms to the Laplace operator $\Delta$. 
In this new frame we consider polar coordinates $(r,\phi)$ centered in $y_1$.
Let $\phi_1,\phi_2$ correspond to $\Gamma_1,\Gamma_2$; 
we normalize $\phi_1\in[0,360^\circ)$ and $\phi_2\in[\phi_1,\phi_1+360^\circ)$.
The coordinate transformation from $(x,y)$ to $(t,\phi)$ with $t=\log r$ is conformal, 
hence preserves the Laplace operator, mapping the cone 
$\{(r,\phi):r>0,~\phi\in(\phi_1,\phi_2)\}$ to an infinite strip $\R\times(\phi_1,\phi_2)$. The 
leading-order parts of $B_k(y_1)$ are
\begin{alignat}{1}
    \frac{\partial u}{\partial t}\cos\gamma_k+\frac{\partial u}{\partial\phi} \sin\gamma_k
    \myeqlabel{eq:beta}
\end{alignat}
Here $\gamma_k$ is the counterclockwise angle from $\Gamma_k$ to the coefficient vector $(g_{k1},g_{k2})$ 
on the corresponding boundary (see Figure \myref{fig:corner}).
We normalize $\gamma_1\in[0,180^\circ)$ and $\gamma_2\in(\gamma_1-180^\circ,\gamma_1]$.

Apply the Fourier-Mellin transform in $t$ to the homogeneous corner equation 
$$-\Delta_{(t,\phi)}u=(-i\partial_t)^2+(-i\partial_\phi)^2=0$$
to obtain the \defm{operator pencil} equation
$$(-i\partial_\phi)^2\tilde u+\lambda^2\tilde u=0$$
where $\lambda=\alpha+i\beta$ are the eigenvalues. The eigenfunctions yield well-known harmonic functions
$$u(t,\phi)=\exp(\beta t)\sin(\beta\phi-\delta)=r^\beta\sin(\beta\phi-\delta).$$
Imposing homogeneous boundary conditions restricts this family to
$$u(r,\phi)=r^{\beta_\ell}\sin(\beta_\ell(\phi-\phi_1)-\gamma_1)$$
with
\begin{alignat}{1}
    \beta_0=-\frac{\gamma_2-\gamma_1}{\phi_2-\phi_1},\qquad\beta_\ell=\beta_0+\frac{\pi}{\phi_2-\phi_1}\ell\qquad(\ell\in\Z). \myeqlabel{eq:alpha}
\end{alignat}
The \defm{multiplicity} of each eigenfunction is $1$, except when $\beta=0$ where it is $2$ 
(for example in the case of two Neumann conditions there is another eigenfunction $u=t=\log r$).

\subsection{Fredholm index jump}

\begin{proposition}
    \mylabel{prop:index-pm}%
    Consider Banach spaces $X_+\subset X_-$ and $Y_+\subset Y_-$ and 
    Fredholm operators $A_\pm:X_\pm\rightarrow Y_\pm$.
    Let $Z:=\{u\in X_-:Au\in Y_+\}$.
    If
    $$d:=\dim(Z/X_+)<\infty,$$
    then
    $$\find A_--\find A_+=d.$$
\end{proposition}
\begin{proof}
    $\ran A_+\subset\ran A_-$, so $\ker A_-^*=(\ran A_-)^\perp\subset(\ran A_+)^\perp=\ker A_+^*$;
    both spaces are finite-dimensional by Fredholmness of $A_\pm$.
    Choose a basis $\psi_1,\dotsc,\psi_r$ for $\ker A_+^*\subset Y_+^*$ so that $\psi_{m+1},\dotsc,\psi_r$ form
    a basis for $\ker A_-^*$. Choose $w_1,\dotsc,w_r\in Y_+$ biorthogonal to $\psi_1,\dotsc,\psi_r$.
    Then $w_1,\dotsc,w_m\in(\ker A_-^*)^\perp=\ran A_-$ by choice of $m$, so we can find $u_1,\dotsc,u_m\in X_-$ with
    $Au_j=w_j\in Y_+$, which also means $u_1,\dotsc,u_m\in Z$ by definition of $Z$.

    Claim: $u_1,\dotsc,u_m,\ker A_-$ are independent modulo $X_+$.
    If not, we could find nontrivial coefficients $\alpha_1,\dotsc,\alpha_m$ as well as 
    $k\in\ker A_-$, $x\in X_+$, so that 
    $$\sum_{i=1}^m\alpha_iu_i=k+x.$$
    Then
    $$A\sum_{i=1}^m\alpha_iu_i=Ak+Ax=Ax,$$
    so
    $$\alpha_j=\psi_j(\sum_{i=1}^m\alpha_iw_i)=\psi_j(A\sum_{i=1}^m\alpha_iu_i)=\psi_j(Ax)=0\qquad(j=1,\dotsc,n)$$
    since $\psi_j\in\ker A_+^*=(\ran A_+)^\perp$ and $Ax\in\ran A_+$. 
    The coefficients are trivial --- contradiction.

    Assume we can add a $u_{m+1}\in Z$ so that
    $u_1,\dotsc,u_{m+1},\ker A_-$ are still independent modulo $X_+$.
    The system
    $$\psi_j(A\sum_{i=1}^{m+1}\alpha_iu_i)=0\qquad(j=1,\dotsc,m)$$ 
    is underdetermined, so we can find a nontrivial solution $\alpha_1,\dotsc,\alpha_{m+1}$.
    Combined with the same result for $j=m+1,\dotsc,r$ (trivial) and with
    $$\psi(A\sum_{i=1}^{m+1}\alpha_iu_i)=0\qquad\text{for $\psi\in\ker A_-^*\subset\ker A_+^*$},$$
    we obtain
    $$\psi(A\sum_{i=1}^{m+1}\alpha_iu_i)=0\qquad\text{for all $\psi\in\ker A_+^*$},$$
    i.e.
    $$A\sum_{i=1}^{m+1}\alpha_iu_i\in\ker(A_+^*)^\perp=\ran A_+,$$
    but then 
    $$A(\sum_{i=1}^{m+1}\alpha_iu_i-d)=0$$
    for some $d\in X_+$, so $u_1,\dotsc,u_{m+1},\ker A_-$ are dependent
    modulo $X_+$ --- contradiction. 

    Hence $u_1,\dotsc,u_m$ form the \emph{basis} of a complement of $\ker A_-$ in $Z$ modulo $X_+$, 
    so\footnote{We write $V/W:=V/(V\cap W)$ for simplicity.}
    $$\dim(Z/X_+)=\dim(\ker A_-/X_+) + m.$$

    Finally,
    \begin{alignat*}{1}
        &\find A_--\find A_+\\
        &= (\dim\ker A_--\dim\ker A_-^*)-(\dim\ker A_+-\dim\ker A_+^*) \\
        &= (\dim\ker A_--\dim\ker A_+)+(\dim\ker A_+^*-\dim\ker A_-^*) \\
        &= \dim(\ker A_-/\dim\ker A_+) + m = \dim(\ker A_-/X_+) + m \\
        &= \dim(\ker A_-/X_+) + \dim(Z/X_+)-\dim(\ker A_-/X_+) = \dim(Z/X_+).
    \end{alignat*}
\end{proof}

\end{document}